\documentclass[conference, letterpaper]{IEEEtran}
\IEEEoverridecommandlockouts
\usepackage{cite}
\usepackage{amsmath,amssymb,amsfonts}
\usepackage{algorithmic}
\usepackage{graphicx}
\usepackage{textcomp}
\usepackage{xcolor}
\def\BibTeX{{\rm B\kern-.05em{\sc i\kern-.025em b}\kern-.08em
    T\kern-.1667em\lower.7ex\hbox{E}\kern-.125emX}}

\usepackage{amssymb}

\usepackage{graphicx}
\usepackage{url}
\usepackage{epsfig}
\usepackage{graphicx}
\usepackage{amsfonts}
\usepackage{amssymb}
\usepackage{amsbsy}
\usepackage{amsfonts}
\usepackage{amsthm, mathrsfs}
\usepackage{longtable}
\usepackage{stackengine}
\usepackage{mathtools}
\usepackage{stmaryrd}
\usepackage{relsize}
\usepackage{multirow}
\usepackage{bm}
\usepackage[T1]{fontenc}
\usepackage{url}
\usepackage{subcaption}
\usepackage{caption}
\usepackage{dsfont}
\usepackage{cite}

\usepackage{color}                
\usepackage{xfrac}          

\usepackage{amssymb,amsmath,amsthm,enumitem}

\theoremstyle{definition}
\newtheorem{theorem}{Theorem}
\newtheorem{proposition}{Proposition}

\begin{document}

\title{Single-Component Privacy Guarantees in Helper~Data~Systems and\\Sparse Coding with Ambiguation}

 \author{%
   \IEEEauthorblockN{Behrooz~Razeghi\IEEEauthorrefmark{1},
   				   Taras~Stanko\IEEEauthorrefmark{2}, 
                     Boris~\v{S}kori\'{c}\IEEEauthorrefmark{2},
                     Slava~Voloshynovskiy\IEEEauthorrefmark{1}
                     \thanks{B.\,Razeghi has been supported by the ERA-Net project
ID\_IoT No 20CH21\_167534 and T.\,Stanko by NWO project Espresso (628.001.019).}
                     }
   \IEEEauthorblockA{\IEEEauthorrefmark{1}%
                     Department of Computer Science, University of Geneva, Switzerland\\
                    \{\texttt{behrooz.razeghi, svolos}\}@\texttt{unige.ch}}
   \IEEEauthorblockA{\IEEEauthorrefmark{2}%
                     Department of Mathematics and Computer Science, Eindhoven University of Technology, the Netherlands\\
                     \{\texttt{t.stanko, b.skoric}\}@\texttt{tue.nl}}
 }

\maketitle

\begin{abstract}
We investigate the privacy of
two approaches to (biometric) template protection:
Helper Data Systems and Sparse Ternary Coding with Ambiguation.
In particular, we focus on a privacy property that is often overlooked,
namely how much leakage exists about one specific binary property of
one component of the feature vector. 
This property is e.g. the sign or an indicator that a threshold is exceeded.

We provide evidence that both approaches are able to protect such sensitive binary variables,
and discuss how system parameters need to be set.

\end{abstract}

\begin{IEEEkeywords}
privacy, biometric authentication, template protection
\end{IEEEkeywords}

\setlength{\parindent}{0mm}


\section{Introduction}
\label{Sec:Introduction}

\subsection{Privacy-preserving storage of biometric enrollment data}
\label{sec:introprivacy}

\vspace{-3pt}

Biometric data such as fingerprints and irises cannot be treated as a secret.
After all, we leave latent fingerprints on many objects that we touch, and 
high-resolution photos of faces reveal a lot about our irises.
Nonetheless, person authentication based on biometrics is still possible, provided
that the verifier performs good liveness detection.
In spite of the not-really-secret nature of biometric data
there are very good reasons to treat them as confidential.
Storing biometric databases in unprotected form would lead to various privacy issues.
In this paper we focus on one particular privacy problem:
some biometric data reveal medical conditions.

The protection of this kind of data must be as good as the protection of passwords.
The attacker model in the case of password storage states that the adversary is an {\em insider}, 
i.e., somebody who has access to cryptographic keys.
Furthermore, the standard use case considered in most of the literature dictates that
the biometric prover does not have to type long keys or to present a smartcard.
This combination of attacker model and use case implies that simply encrypting the confidential data is not an option.
The typical solution for passwords is to apply a one-way function and to store
the hash of each password.
However, this solution does not work for noisy data such as biometrics;
one bit flip in the input of the hash function causes 50\% bit flips in the output.

Several techniques have been developed for securely storing noisy credentials, also known as template protection, in the
above given context:
(i) Helper Data Systems (HDS), also known as fuzzy commitment, secure sketch, fuzzy extractor
\cite{JW99,LT2003,DRS2004,DORS2008};
(ii) Locality Sensitive Hash (LSH) functions \cite{indyk1998approximate, datar2004locality}; 
(iii) homomorphic encryption \cite{lagendijk2012encrypted, aguilar2013recent}; and most recently
(iv) Sparse Coding with Ambiguation (SCA) \cite{Razeghi2017wifs, Razeghi2018icassp, Razeghi2018eusipco, Razeghi2019icip}.

\subsection{Comparison of Template Protection Techniques}
\label{sec:introcomparison}

\vspace{-3pt}

The LSH approach is fast but does not give clear privacy guarantees.
Homomorphic encryption has excellent privacy, but is computationally expensive.
In this paper we will not consider the LSH and homomorphic crypto approach.

The HDS approach is the oldest and is well studied. 
Nevertheless, the narrow privacy question of {\em protecting one specific aspect} of the biometric, which is relevant for the above mentioned medical condition,
has not been studied in detail.

The aim of this paper is to compare the privacy properties of the HDS and the SCA approach,
in particular the `medical condition' aspect.
Here it is important to note that previous work on SCA has focused only on the
inability of an adversary to reconstruct the {\em full biometric} from the enrollment data;
that is {\em not} the property we will be looking at in the current paper.
Mostly, in the literature, the protection of a vector $\mathbf{x} \in \mathbb{R}^N$ is considered. 
However, often it is the projection of $\mathbf x$ onto some fixed direction $\mathbf v$ that is relevant,
$z=\mathbf x\cdot\mathbf v$.
The range in which $z$ lies can be privacy-sensitive, e.g. the sign of $z$ or whether $z$ is far away from average.

\subsection{Contributions}
\label{sec:contributions}

\vspace{-3pt}

We concentrate on one component $x_n$ of a to-be-protected random vector $\mathbf x$,
in particular a binary property $\psi(x_n)$, which is either the sign or an `extremeness' indicator
that checks if $|x_n|$ exceeds some threshold.
We investigate how much information leaks about $\psi(x_n)$ through the enrollment data.
\begin{itemize}[leftmargin=*]
\item
In quantizing HDSs high leakage can occur if a bad parameter choice is made.
The best choice is to take an {\em even} number of quantization intervals, and to 
subdivide them into {\em two} helper data intervals;
then there is zero leakage about the sign and the `extremeness'.
\item
The Code Offset Method causes negligible leakage.
\item
In the SCA mechanism,
leakage about the `extremeness' bit can be made negligibly small by setting
the ambiguation noise level larger than the Hamming weight of the sparse ternary representation.
This noise has little impact on the performance of the authentication system,
since it gets removed in case of a genuine user's verification measurement.
\end{itemize}


\section{Preliminaries}
\label{sec:preliminaries}

\vspace{-2pt}

\subsection{Notation and Terminology}
\label{sec:notation}

\vspace{-2pt}

Vectors and matrices are denoted by boldface lower-case ($\mathbf{v}$) and upper-case ($\mathbf{M}$). 
When there is no distinction between a scalar ($w$), vector ($\mathbf{w}$), or matrix ($\mathbf{W}$), 
we write $\mathsf{w}$. 
The notation $\mathbf{m}{\left(  j \right)}$ denotes the $j$-th column of $\mathbf{M}$. 
An enrollment measurement will be written as a vector $\mathbf x$,
and the verification measurement as~$\mathbf y$.
When a distinction needs to be made between a random variable (RV) and its numerical value,
the RV is written in capitals, and the value in lowercase.
Expectation over $x$ is denoted as $\mathbb E_x$.
The Shannon entropy of a discrete RV $X$ is denoted as $H(X)$ and is defined as
$H(X)\!=\!\sum_x p_x\log\frac1{p_x}$.
The conditional entropy of $X$ given $Y$ is written as $H(X|Y)$ and is defined as
$H(X|Y)\!=\!\mathbb E_y H(X|Y\!=\!y)$.
The mutual information between $X$ and $Y$ is
$I(X;Y)\!=\!H(X)-H(X|Y)$.
For $p\in[0,1]$ the binary entropy function $h$ is defined as
$h(p)=p\log \frac1p+(1-p)\log\frac1{1-p}$.
We use the notation $\left[ n \right]=\{   1,  ..., n   \}$. 
The superscript $(\cdot)^{\!T}$ stands for the transpose and $(\cdot)^{\! \dagger}$ for pseudo-inverse. 
The logarithm `log' has base~$2$.
Bitwise XOR is denoted as $\oplus$.
The Heaviside step function is written as $\Theta(\cdot)$.

We will work with the following {\em authentication} setting.
The Verifier owns an enrolment database of public data $P(c)$
for a set of users $c\in [C]$. In the verification phase he is presented with a vector
$\mathbf y$ and a user label $c$; his task is to decide if $\mathbf y$ is consistent 
with $P(c)$.

\begin{figure}[t]
    \centering
     \hspace{-8pt}
        \begin{subfigure}[h]{0.24\textwidth}
        \includegraphics[scale=0.385]{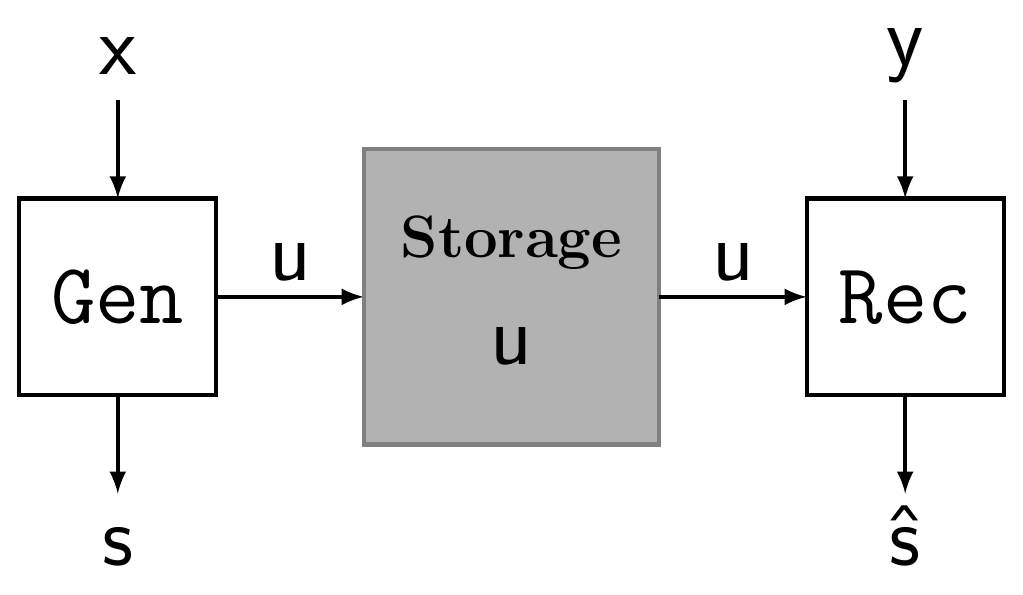}%
            \vspace{-10pt}
        \caption{ }
            \vspace{-2pt}
        \label{fig:HDS}
    \end{subfigure}%
~
       \begin{subfigure}[h]{0.24\textwidth}
        \includegraphics[scale=0.385]{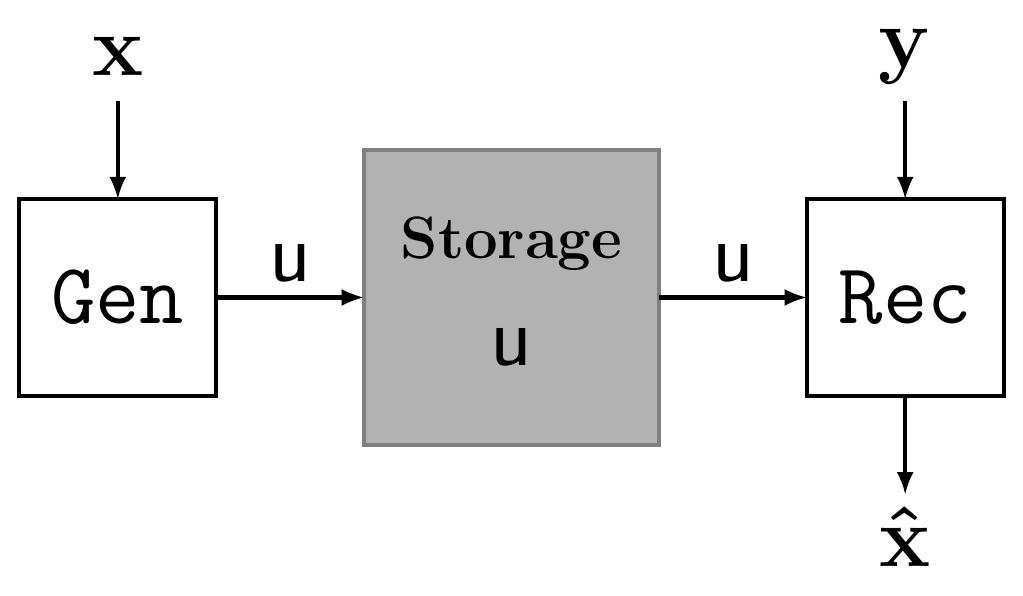}%
         \vspace{-10pt}
        \caption{ }
            \vspace{-2pt}
        \label{fig:SCA}
    \end{subfigure}
    \caption{\it Data flow in: (a) generic Helper Data System and (b) general Sparse Coding with Ambiguation mechanism.}
    \label{Fig:DataFlow_HDS_SCA}
\end{figure}

\subsection{Zero Leakage Helper Data Systems}
\label{sec:ZLHDS}

\vspace{-2pt}

A HDS in its most general form is shown in Fig.\,\ref{fig:HDS}. 
The {\tt Gen} procedure takes as input a measurement $\mathsf{x}$. 
{\tt Gen} outputs a secret $\mathsf{s}$ and public Helper Data $\mathsf{u}$, where $\mathsf{x}$, $\mathsf{s}$ and $\mathsf{u}$ can be scalar, vector, or matrix, in general. 
The helper data is stored in a memory that can be accessed by the adversary. 
In the reproduction phase, a fresh measurement $\mathsf{y}$ is obtained. 
Typically $\mathsf{y}$ is close to $\mathsf{x}$ but not identical. 
The {\tt Rec} procedure takes $\mathsf{y}$ and $\mathsf{u}$ as input.
It outputs $\mathsf{\hat{s}}$, an estimate of $\mathsf{s}$. 
If $\mathsf{y}$ is sufficiently close to $\mathsf{x}$, then $\mathsf{\hat{s}}= \mathsf{s}$. 
The `Zero Leakage' (ZL) property is defined as $I(U;S)=0$,
i.e.,\,the helper data reveals nothing about the secret.
Obviously, $U$ has to leak about $X$, since $U$ is a function of~$X$.
When $ \mathsf{s} = \mathsf{x}$ the HDS is referred to as a Secure Sketch. 
When $S$ given $U$ has a uniform distribution, the HDS is called a Fuzzy Extractor.
If $X$ is a continuum variable, the first step in the signal processing is discretisation.
For this purpose, a special ZLHDS has been designed 
\cite{VTOSS10,dGSdVL,SAS2016}
which reduces quantisation errors. 
The distribution of $X$ needs to be known.
After discretisation the Code Offset Method can be applied (Section~\ref{sec:COM}).

The discretising ZLHDS is shown in Fig.\,\ref{fig:stripes}.
Consider a source $X\in\mathbb R$, with $X\sim f(x)$.
Let $F$ be the cumulative distribution.
The $x$-axis is divided into $J$ quantisation intervals corresponding to the
extracted secret $s\in \{0, 1, \ldots, J-1 \}$.
The distribution of $S$ is not necessarily uniform.
Let $q_s$ be the left boundary of the interval $S=s$.
Let $\{ p_s\}_{s=0}^{J-1}$ denote the probabilities of the $s$-values.
Then $F(q_s)=\sum_{i=0}^{s-1}p_i$.
Each $s$-interval is equiprobably divided into $m$ sub-intervals
(depicted as grayscales in Fig.\,\ref{fig:stripes}); the helper data $u$ is defined as the index of the sub-interval in which the enrollment $x$ lies.
The $s$ and $u$ are computed from $x$ as follows:\vspace{-3pt}
\begin{eqnarray}
    s=\max \{t| q_t\leq x\}
    &;&
    u=\left\lfloor m\frac{F(x)-F(q_s)}{p_s} \right\rfloor.
\end{eqnarray}
In the limit $m\to\infty$, the helper data can be seen as a {\em quantile} 
$\tilde u=\frac{u}{m}\in[0,1)$ within the $s$-interval.
It holds that
$F(x) = F(q_s) + \tilde{u} p_s$.

\begin{figure}[!t]
\centering
\includegraphics[scale=0.28]{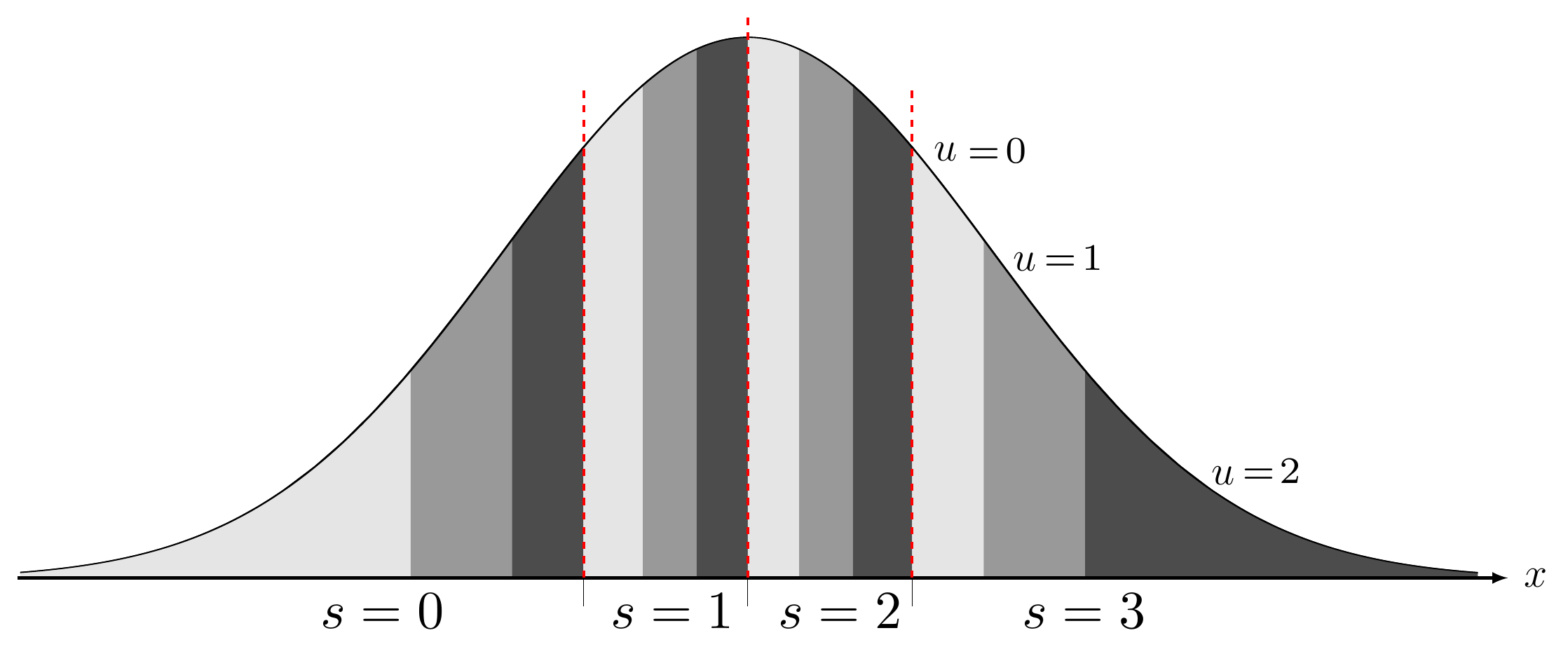}
\caption{\it Example of the Zero Leakage discretising HDS with four quantisation intervals $s\in \{0,1,2,3\}$.
The discrete helper data $u\in  \{0,1,2\}$ is indicated as grayscales.}
\label{fig:stripes}
\end{figure}

\subsection{The Code Offset Method \cite{BBCS1991,JW99} }
\label{sec:COM}

\vspace{-2pt}

Consider $\psi (\mathbf{x}), \psi (\mathbf{y}) \in\{0,1\}^N$. 
Consider a linear binary error-correcting code with syndrome function ${\tt Syn}:\{0,1\}^N \to\{0,1\}^{N-k}$
and syndrome decoding function
${\tt SynDec}:\{0,1\}^{N-k}\to\{0,1\}^N$, where $k$ is the message length. 
The Code Offset Method in its simplest form can be used as a Secure Sketch for any distribution of~$X$.
The helper data and the reconstruction $\widehat{\psi (\mathbf{x})}$ are defined as follows:\vspace{-3pt}
\begin{eqnarray}
    \mathbf{u}={\tt Syn}\, \psi (\mathbf{x}); \;\;
   \widehat{\psi ( \mathbf{x})} = \psi (\mathbf{y}) \oplus {\tt SynDec}(\mathbf{u} \oplus{\tt Syn}\,\psi (\mathbf{y})).\!\!
\end{eqnarray}


\subsection{Sparse Ternary Coding (STC) \cite{Sohrab_ISIT2017, Razeghi2018icassp}}
\label{ssec:SparseTernaryCoding}

\vspace{-3pt}

The encoder is a mapping $\varphi \! :  \! \mathbb{R}^N  \! \rightarrow \! \{ -1, 0 , +1 \}^L$, 
where $L$ may be smaller, equal to, or larger than $N$. 
The $\varphi(\mathbf x)$ is a {\em sparse encoding} of $\mathbf x$; the number of nonzero entries is $S_t\ll L$,
which is called the sparsity level. 
The {\bf encoder} first applies a projection matrix $\mathbf W\in{\mathbb R}^{L\times N}$ and then element-wise thresholding 
$\psi_\lambda^{\rm stc}$, where $\lambda$ is a parameter,
$\psi_\lambda^{\rm stc}(q)\triangleq {\rm sign}(q)\Theta(|q|-\lambda)$ (see Fig.\,\ref{fig:TernaryOperator}).
\begin{equation}
\label{Eq:STC}
\mathbf{v}  \! = \varphi \left(\mathbf{x} \right) =   \psi_{\lambda}^{\mathrm{stc}} \left( \mathbf{W} \mathbf{x}  \right) \! \in \! \{ -1, 0 , +1 \}^L.
\end{equation}
The threshold $\lambda$ is tuned to get the desired sparsity level $S_t$.
The {\bf decoder} produces an estimator for $\mathbf x$ as $\widehat{\mathbf x}={\mathbf W}^\dag \mathbf v$.


\subsection{Sparse Binary Coding (SBC)}
\label{ssec:SparseBinaryCoding}

Here the thresholding function is $\psi_\tau^{\rm sbc}(x)\triangleq \Theta(|x|-\tau)$ 
(Fig.\,\ref{fig:SparseBinaryOperator}).  
Given a (raw) feature vector $\mathbf{x} \in \mathbb{R}^N$ SBC generates a binary vector 
$\psi_{\tau}^{\mathrm{sbc}}  \left(\mathbf{x}  \right) \in \{ 0 ,1 \}^N$.


\subsection{Binary Coding (BC)}
\label{ssec:BinaryCoding}

This is the component-wise sign operation (Fig.~\ref{fig:BinaryOperator}) excluding zeros.
Given a (raw) feature vector $\mathbf{x} \in \mathbb{R}^N$ the BC simply generates a binary vector $\psi^{\mathrm{bc}}  \left(\mathbf{x}  \right) = \mathrm{sign} \left(  \mathbf{x} \right)\in \{ -1 ,1 \}^N$.

\subsection{Sparse Coding with Amibiguation (SCA) \cite{Razeghi2017wifs, Razeghi2018icassp}}
\label{sec:ambigu}

Given a sparse ternary vector $\mathbf v=\varphi (\mathbf{x})\in\{-1,0,1\}^L$, 
the ambiguation mechanism $A$ turns $S_n$ randomly chosen
zero components of $\mathbf v$ into a (random) $\pm1$.
The resulting ternary vector $\mathbf u=A(\mathbf v)$ is stored as enrolment data (see Fig.~\ref{Fig:Enrollment}),
together with the matrix $\mathbf W$.
The randomly added nonzero components make it prohibitively difficult to reconstruct 
$\mathbf x$ from $\mathbf u$, while still allowing a verifier to check if a verification measurement $\mathbf y$ is consistent with $\mathbf u$.

\begin{figure}[!t] 
    \centering%
    \begin{subfigure}[b]{0.145\textwidth}
        \includegraphics[scale=0.55]{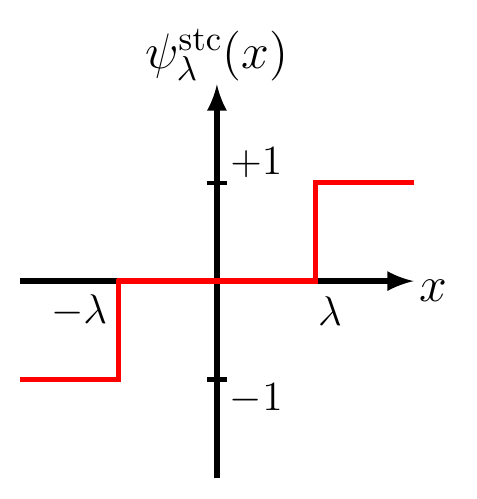}
         \vspace{-9pt}
         \caption{ }
         \label{fig:TernaryOperator}
    \end{subfigure}
     ~ 
    \begin{subfigure}[b]{0.145\textwidth}
        \includegraphics[scale=0.55]{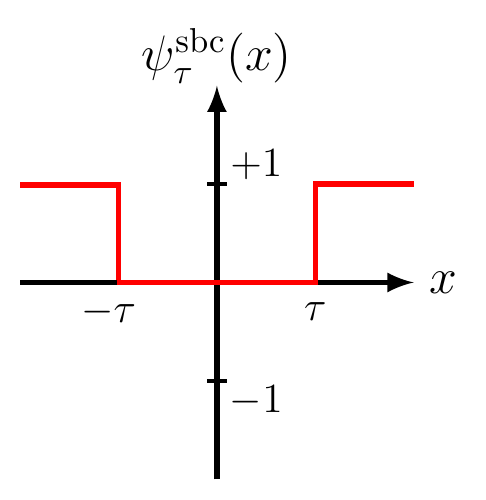}
         \vspace{-9pt}
         \caption{ }
        \label{fig:SparseBinaryOperator}
    \end{subfigure}
    ~ 
    \begin{subfigure}[b]{0.145\textwidth}
        \includegraphics[scale=0.55]{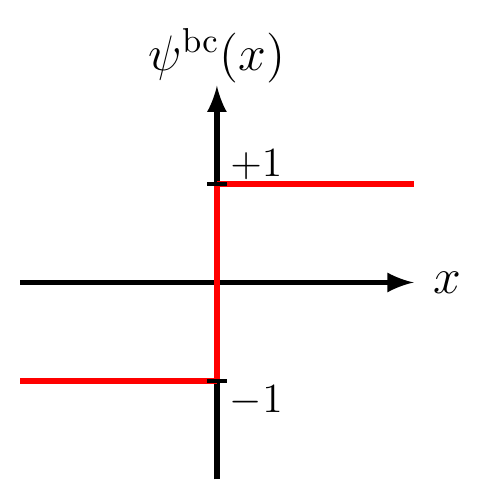}
        \vspace{-9pt}
         \caption{ }
        \label{fig:BinaryOperator}
    \end{subfigure}
    \caption{\it (a) Ternary thresholding; (b) Binary thresholding; (c) Binarisation.}
 \label{Fig:SparsifyingOperators}
\end{figure}

\section{Problem formulation}
\label{sec:ProblemFormulation}

We formally capture the `medical condition' issue as follows.
We model the existence of the privacy-sensitive medical condition as a binary function of one of the components of the enrollment measurement $\mathbf x \in \mathbb{R}^N$.
That is, we say that $\psi (x_n)$ is the quantity that should not leak, for some $n\in[N]$.

We will work with two choices for the function $\psi (\cdot)$ that seem to make sense in our context: 
(a) $\psi (\cdot) = \psi^{\mathrm{sbc}}_{\tau} (\cdot) $ (see Fig.~\ref{fig:SparseBinaryOperator}), and 
(b) $\psi (\cdot) = \psi^{\mathrm{bc}} (\cdot)$ (see Fig.~\ref{fig:BinaryOperator}). 
We will assume that the index $n$ is not known to the legitimate parties at the time of enrollment. Otherwise, there is a trivial solution.

We will consider only distributions of $X$ that are symmetric around $X=0$,
i.e., even functions~$f(x)$.
Furthermore we work in the `perfect enrollment' model, which states that there is
no measurement noise at enrollment time. 
In this way we are erring on the side of caution, overestimating the leakage.

\section{Results for Helper Data Systems}
\label{sec:resultsHDS}

We zoom in on the relevant component $X_n$.
We introduce shorthand notation $Z =  \psi^{\mathrm{sbc}}_{\tau} (X_n) $ %
and $V = \psi^{\mathrm{bc}} (X_n)$.
We write $F$ for the cumulative distribution function of~$X_n$.

\begin{figure}[!t]
\centering
\includegraphics[scale=0.34]{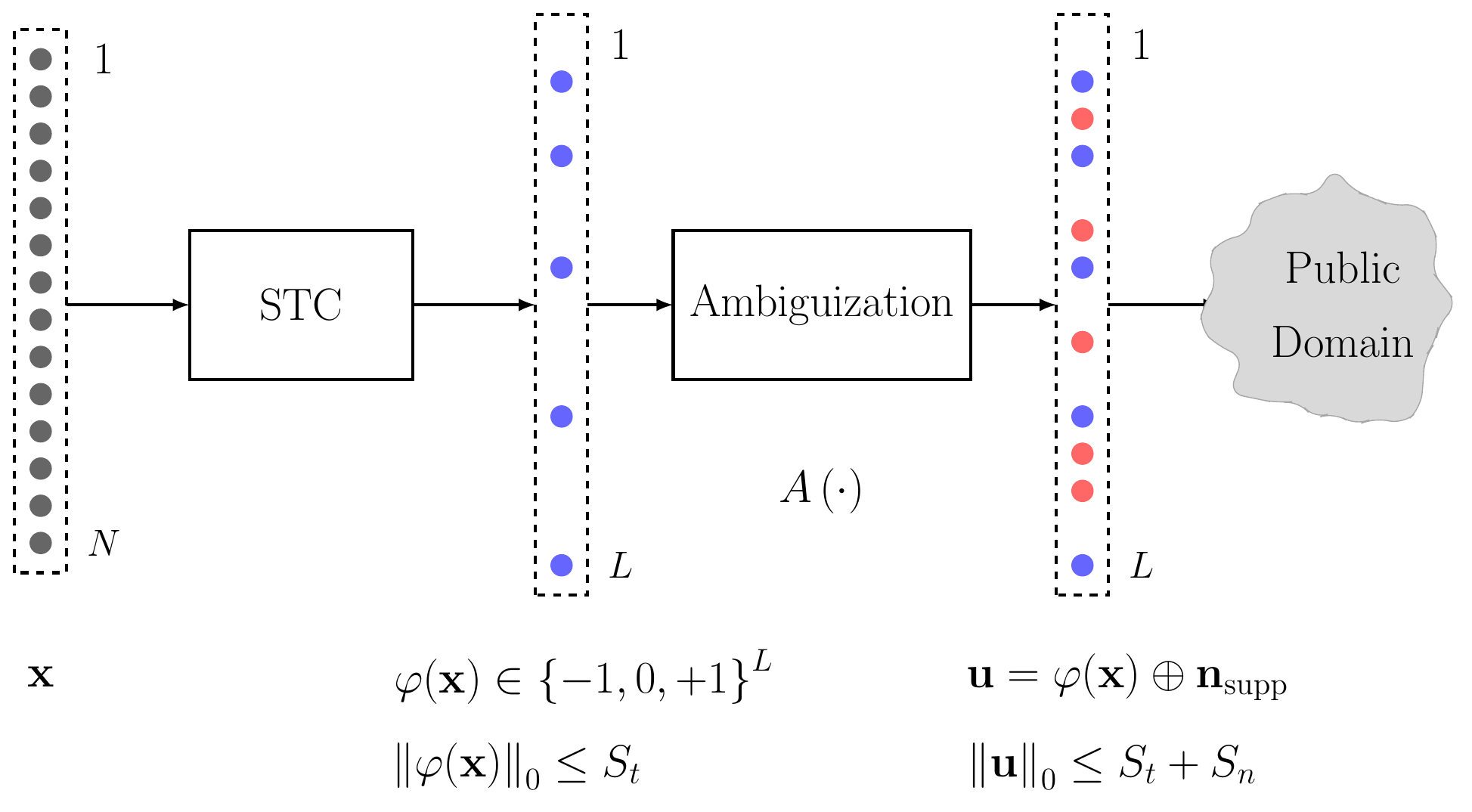}
\caption{\it Enrolment phase of the SCA scheme.}
\label{Fig:Enrollment}
\end{figure}

\subsection{Leakage from the Quantising HDS}
\label{sec:leakHDS1}

First we look at the sign variable~$V$.

\begin{theorem}
\label{th:leakRHDS}
Let the distribution of $X_n$ be an even function. Then $H(V)=1$ and
\begin{eqnarray}
    H(V|U) \! = \! \left\{ \! \!  \! \begin{array}{rl}
    J\mbox{ even}\hfill : & 1 ,\cr
    J=2t+1,\,m\mbox{ even\hfill  :}& h(\frac{1-p_t}2) ,\cr
    J=2t+1,\, m\mbox{ odd \hfill :}&  \frac{m-1}m h(\frac{1-p_t}2)+\frac1m.
    \end{array}\right. \!\!\!\!\! 
\end{eqnarray}
\end{theorem}
\begin{proof}
When $J$ is even, it holds that $\Pr[V=1]=\frac12$ for any~$u$.
When $J$ is odd and $m$ is even, it holds for any $u$ that 
$\Pr[V=1|U=u]=\sum_{n>t}p_n=\frac12\sum_{n\neq t}p_n=\frac12(1-p_t)$.
When $J$ is odd and $m$ is odd, the above situation holds for $u\neq \lfloor m/2\rfloor$,
but for the one special value of $u$ in the middle we have
$\Pr[V=1|U=\lfloor\frac m2\rfloor]=\frac12$.
\end{proof}

Next we look at the threshold indicator~$Z$. 
The entropy of $Z$ is given by $H(Z)=h(F(-\tau)+1-F(\tau))$.
For symmetric $f(x)$ this reduces to $H(Z)=h(2F(-\tau))$.

Theorem~\ref{th:tausmall} below gives an expression for the entropy of $Z$ given
the helper data, in the regime where the threshold $\tau$ lies in the outermost 
$u$-region (the rightmost grayscale band in Fig.\,\ref{fig:stripes}), i.e., the
medical condition is rare.

\begin{theorem}
\label{th:tausmall}
Let $f(x)$ be a symmetric pdf.
Let $F(-\tau)<\frac{p_0}{m}$. Then
\begin{equation}
    H(Z|U)= \frac2m h\left(mF(-\tau)\right).
\end{equation}
\end{theorem}
\begin{proof}
For $u\in\{1,\ldots,m-2\}$ it is certain that $Z=0$.
This gives $H(Z|U)=\frac1m\sum_{w=0}^{m-1}H(Z|U=u)=\frac1m[H(Z|U=0)+H(Z|U=m-1)]$.
Due to symmetry this equals $\frac2m H(Z|U=0)$.
Since $Z$ is a binary RV we have $H(Z|U=0)=h(\Pr[Z=1|U=0])$.
Finally we use $\Pr[Z=1|U=0] = p_0\cdot\frac{F(-\tau)}{p_0/m}$.
\end{proof}

\underline{Remark}:
For $m=2$ we see that $H(Z|U)=H(Z)$, i.e.,\,there is no leakage.

It is interesting to note that {\bf in the HDS with even J and m=2,
the helper data leaks absolutely nothing about V and Z}.

\vskip2mm

One may wonder if $m>2$ is still worth considering.
It should be noted that
the noise resilience of the HDS improves when the number of subdivisions $m$ is increased.
Therefore we cannot exclude that setting $m>2$ can be a good design choice.
Theorem~\ref{th:leakZcontW} below gives a leakage result for the limiting case
$m\to\infty$, i.e.\,the continuum helper data.

\begin{theorem}
\label{th:leakZcontW}
Let $f(x)$ be an even function. Let $F(-\tau)\leq p_0$.
\begin{eqnarray}
    && \!\!\!\!\!\!\!\!\!\!\!\! H(Z|\tilde U)=
    \nonumber\\ &&  \!\!\!\!\!\!\!\!\!\!\!\!
    \left\{\begin{matrix}
    F(-\tau)\leq\frac{p_0}2: &   \frac{2F(-\tau)}{p_0}h(p_0)
    \cr 
    F(-\tau)>\frac{p_0}2: & 
    \frac{2p_0-2F(-\tau)}{p_0}h(p_0)+\frac{2F(-\tau)-p_0}{p_0}h(2p_0).
    \end{matrix}\right.
    \quad\quad
\end{eqnarray}
\end{theorem}

\begin{proof}
We compute $H(Z|\tilde U)={\mathbb E}_{\tilde u} H(Z|\tilde U=\tilde u)$.
Since $\tilde u$ is uniform on $[0,1)$ this evaluates to
$\int_0^1 \mathrm{d} \tilde u\; H(Z|\tilde U=\tilde u)=\int_0^1 \mathrm{d} \tilde u\; h(\Pr[Z=1|\tilde U=\tilde u])$.

\underline{Case $F(-\tau)/p_0\leq\frac12$}.
For $p_0 \tilde u\in[F(-\tau,p_0-F(-\tau))]$ it is certain that $Z=0$.
For all other $\tilde u$ we have $\Pr[Z=1|\tilde U=\tilde u]=p_0$.

\underline{Case $F(-\tau)/p_0>\frac12$}.
For $p_0 \tilde u\in[p_0-F(-\tau), F(-\tau)]$ we have $\Pr[Z=1|\tilde U=\tilde u]=2p_0$ (left tail and right tail).
For all other $\tilde u$ the probability is~$p_0$.
\end{proof}

The relative leakage $[H(Z)-H(Z|\tilde U)]/H(Z)$ is plotted in Fig.\,\ref{fig:leakHDS}.
The dependence on $\tau$ looks strange, but three special points can be understood.
(i) For $\tau\to\infty$ the continuum helper data completely reveals~$Z$;
(ii) For $F(-\tau)=p_0$ the helper data contains no information about $|x|$ crossing the threshold.
That information is contained in~$S$, and the ZLHDS has been designed not to leak anything about~$S$;
(iii) At the intermediate value $F(-\tau)/p_0=1/2$ a special symmetry occurs between the left region
$S=0$ and the right region $S=J-1$. 
Due to this symmetry the conditional distribution $Z|\tilde U=\tilde u$ looks the same for every value of $\tilde u$. 

Furthermore the leakage is a decreasing function of~$J$, since at large $J$ the two tail regions
$S=0$ and $S=J-1$ have less influence.

\begin{figure}[!t]
\centering
\includegraphics[scale=0.33]{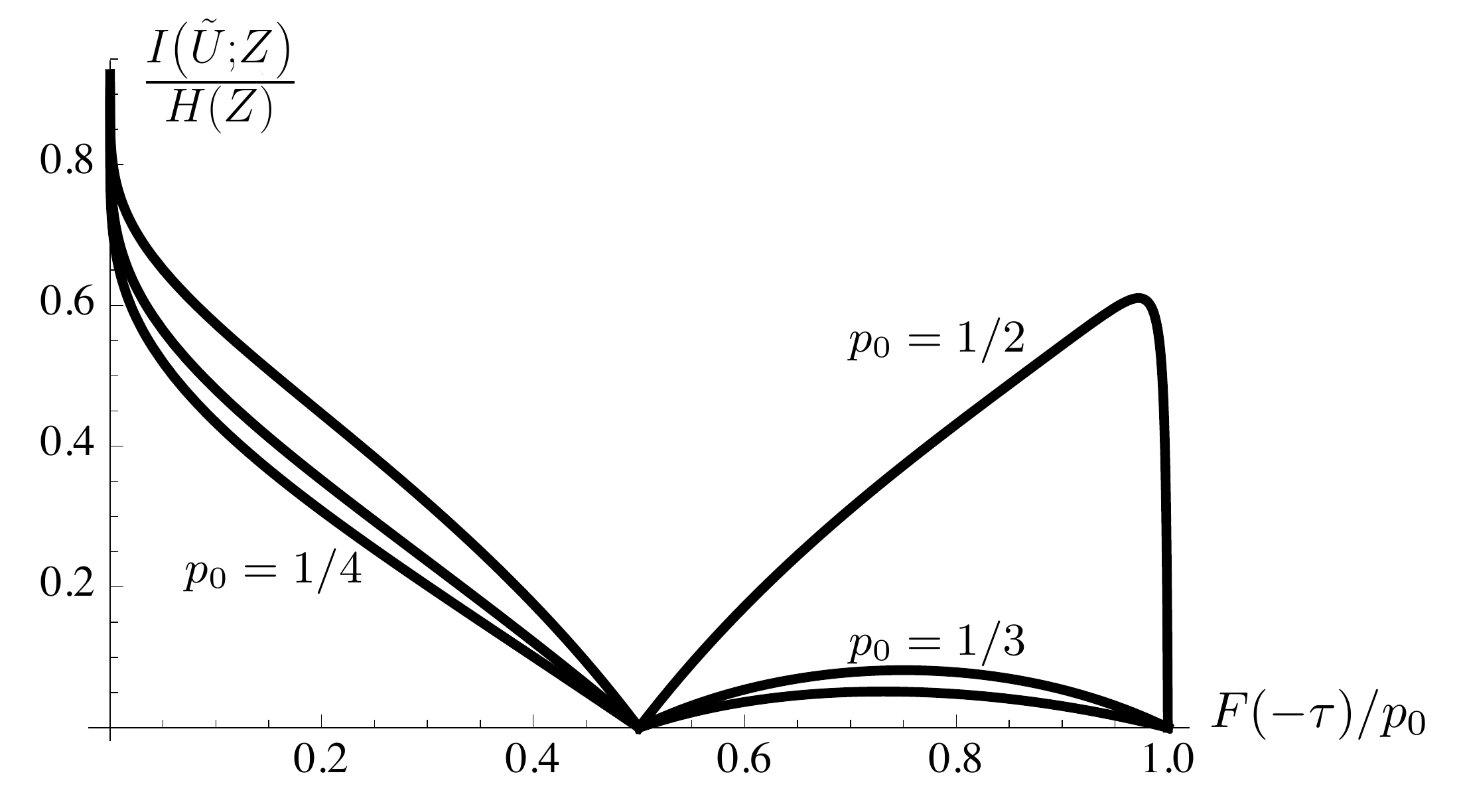}
\caption{\it Normalised leakage $\frac{I(\tilde{U};Z)}{H(Z)}$ as a function of $F(-\tau)/p_0$.
From top to bottom $p_0=\frac12,\frac13,\frac14$.}
\label{fig:leakHDS}
\end{figure}

\subsection{Leakage from the Code Offset Method}
\label{sec:leakCOM}

Even if the quantising HDS does not reveal $V$, the helper data from the Code Offset Method
may still do.
Typically the sign $V\in\{-1,1\}$ gets turned into a bit value somewhere in the binary string
that serves as input to the COM.

We use notation as in Section~\ref{sec:COM}.
Consider {\em uniform} $\psi( X) \in\{0,1\}^N$, and helper data $U={\tt Syn}\, \psi( X)$.
At fixed $u\in\{0,1\}^{N-k}$, there are $2^k$ strings $x$ that are consistent with $u$,
and all of them have equal probability. 
The marginal distribution for one component of $\psi( X)$ is uniform; hence there is no leakage about~$V$.

The above reasoning no longer applies when $\psi(X)$ is non-uniform. 
Then the $2^k$ strings compatible with $u$ are not uniform.
However, the number of strings summed over in the computation of the marginal
is exponentially large; for most `normal-looking' $\psi(X)$-distributions the 
marginal of one component will still be close to uniform.

Finally we briefly depart from the {\em perfect enrollment} setting and investigate the effect of
measurement noise at enrollment time on the privacy properties of the COM.
Suppose there exists a `true' biometric $B\in\{0,1\}^N$, which is shielded from our view by enrollment noise~$G$.
The enrollment measurement yields $B\oplus G$.
Then the relevant privacy question is how much is leaking about (parts of) $B$,
as opposed to~$B\oplus G$.

\begin{proposition}
Let the enrollment noise $G$ be bitwise iid Bernoulli noise with bit error rate $\varepsilon$.
Let $r$ be the row weight of the error correcting code. 
Let $U={\tt Syn}(B\oplus G)$.
Then
\begin{equation}
    I(B;U) \approx 
    (N-k)[1-h\left(\frac12-\frac12(1-2\varepsilon)^r\right)].
\end{equation}
\end{proposition}

\begin{proof}
$I(B;U)=H(U)-H(U|B)=H(U)-H({\tt Syn}\,B \oplus {\tt Syn}\,G |B)$
$=H(U)-H({\tt Syn}\,G)$.
We use $H(U)\leq N-k$.
In a good code the redundancy $N-k$ is just slightly larger than the
entropy of the syndrome ${\tt Syn}\,G$.
We estimate the entropy of ${\tt Syn}\,G$ by setting it close to the Gallager bound \cite{Gallager1963},
$H({\tt Syn}\,G)\approx (N-k)h(\alpha)$, where
$\alpha$ is the bit error probability of $r$ binary symmetric channels concatenated together,
$\alpha=\frac12-\frac12(1-2\varepsilon)^r$.
\end{proof}

For reasonable values of the row weight, we see that the leakage $I(B;U)$
is very small, approximately
$(N-k)\frac{(1-2\varepsilon)^{2r}}{2\ln2}$; 
and this is leakage about the whole vector~$B$.

\section{Results for Sparse Coding with Ambiguation}
\label{sec:resultsSCA}

\begin{figure}[!t]
\centering
\includegraphics[scale=0.3]{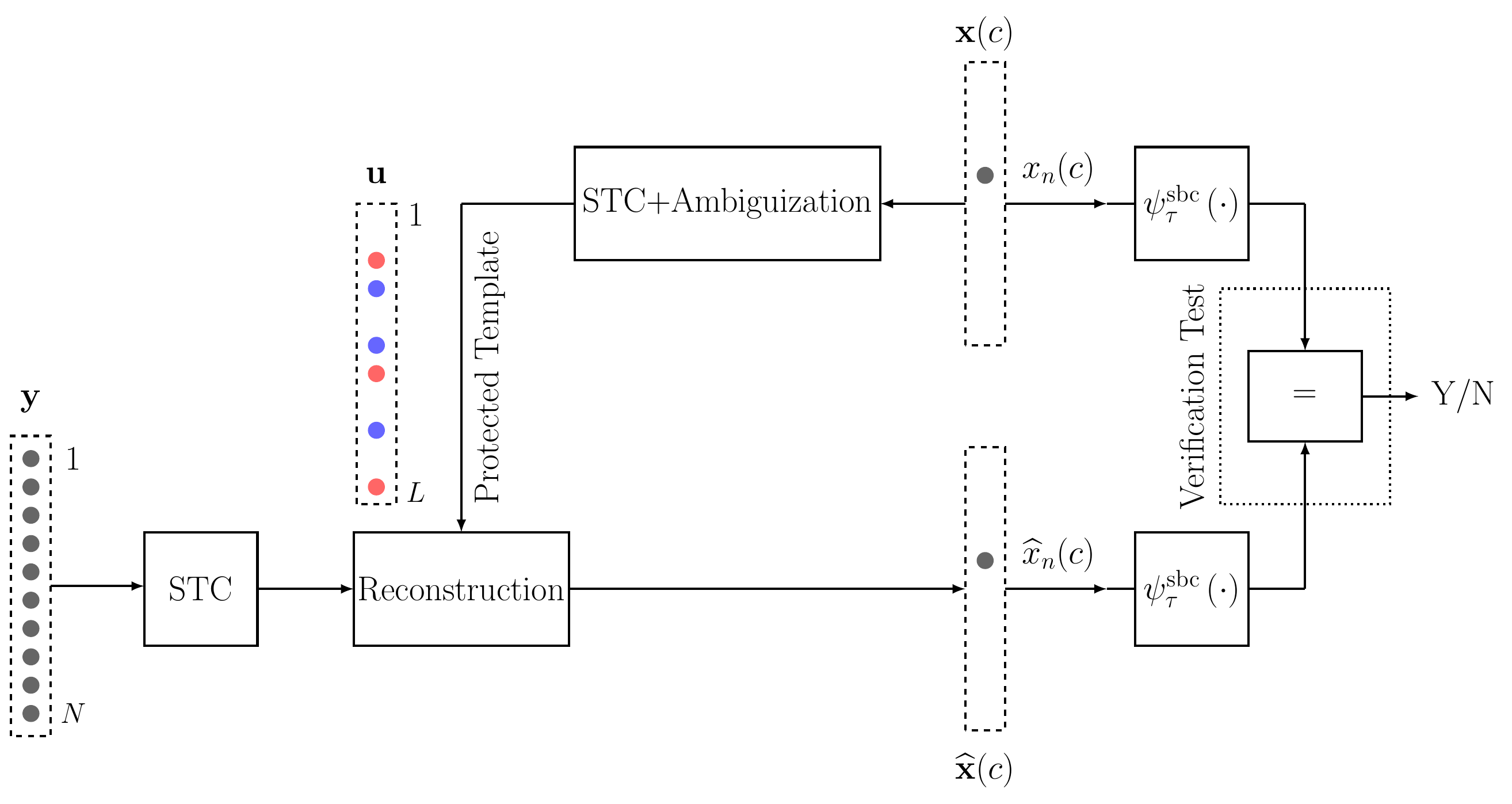}
\caption{\it The schematic block diagram of a physical verification system.}
\label{Fig:TestPhase}
\end{figure}

\subsection{Method}

We consider $C$ users. Each user $c \in [C]$ has a measurement vector $\mathbf x(c)$;
it consists of $N$ components which are modeled as zero-mean unit-variance Gaussian variables.
The enrolment of the vectors $\mathbf{x}(1), \ldots,\mathbf{x}(C)$
is as described in Section~\ref{sec:ambigu}.
The public data for user $c$ is $\mathbf W(c),\mathbf u(c)$.

 \begin{figure}[!t] 
      \centering%
    \begin{subfigure}[b]{0.23\textwidth}
     \includegraphics[width=4.66cm, height=3.55cm]{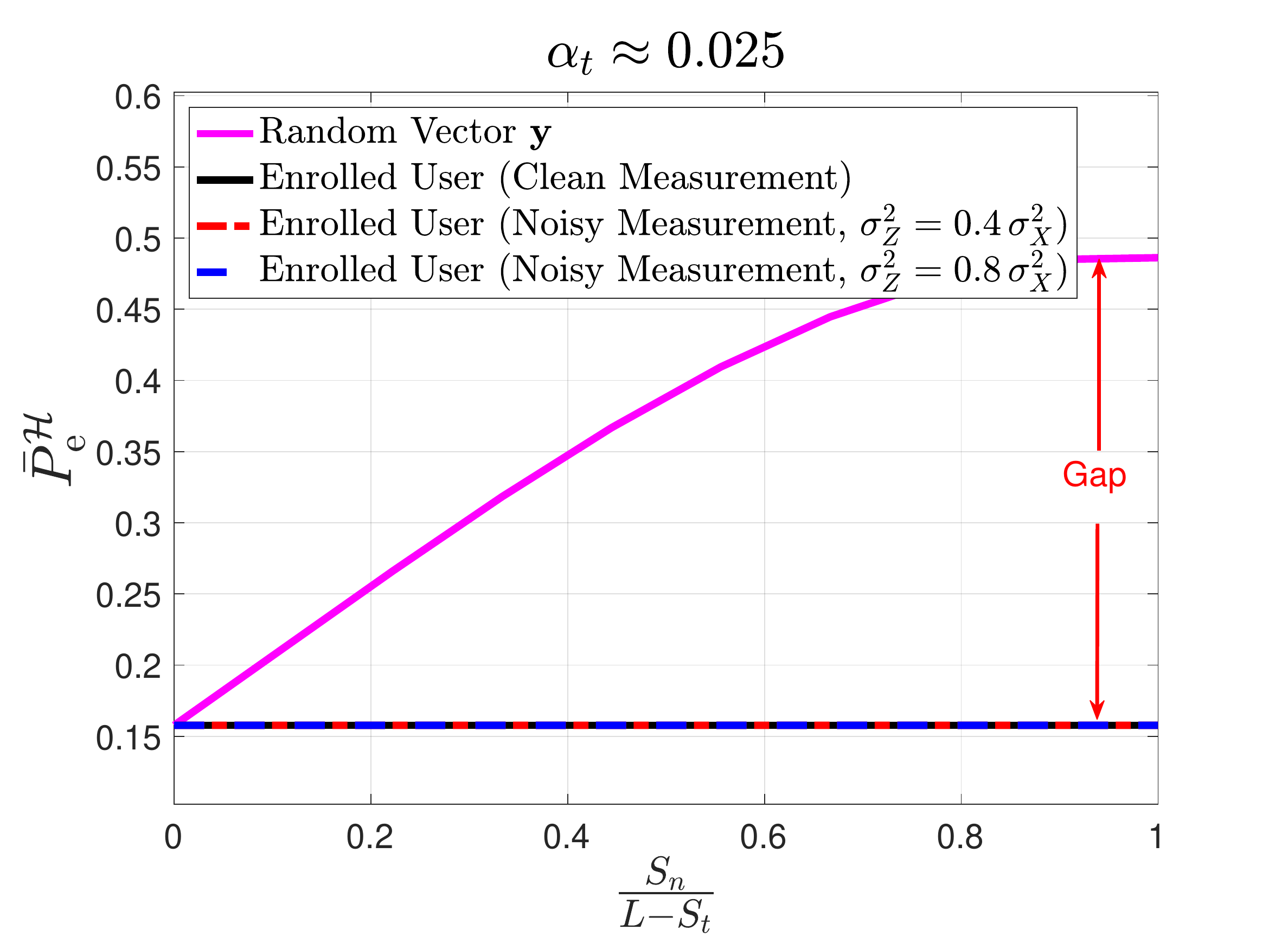}
        \vspace{-17pt}
         \caption{ }
         \label{fig:AveProbebilityError1}
    \end{subfigure}
    ~
    \begin{subfigure}[b]{0.23\textwidth}
        \includegraphics[width=4.65cm, height=3.5cm]{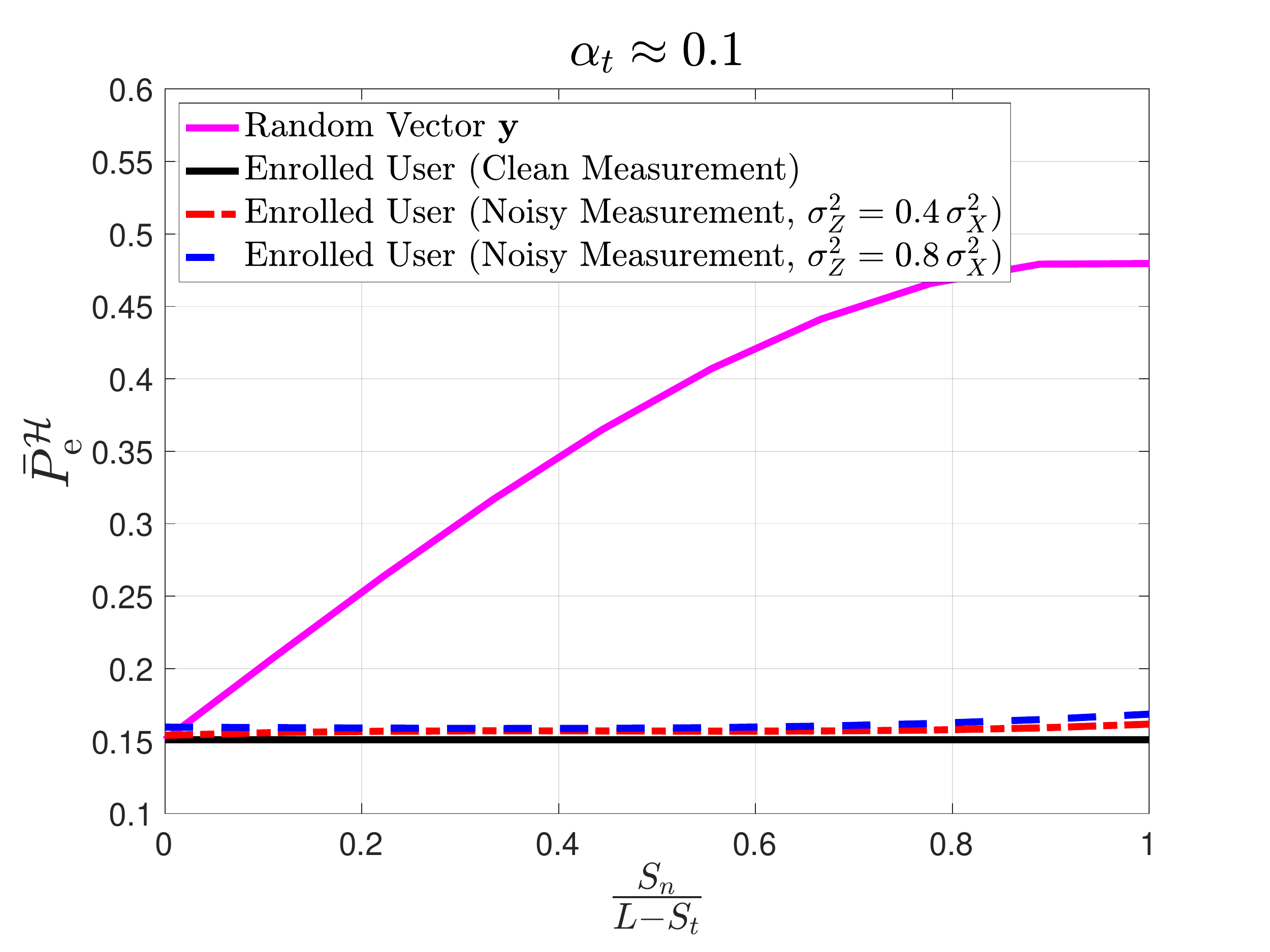}
        \vspace{-17pt}
         \caption{ }
        \label{fig:AveProbebilityError2}
    \end{subfigure}
    
     \begin{subfigure}[b]{0.23\textwidth}
        \includegraphics[width=4.65cm, height=3.5cm]{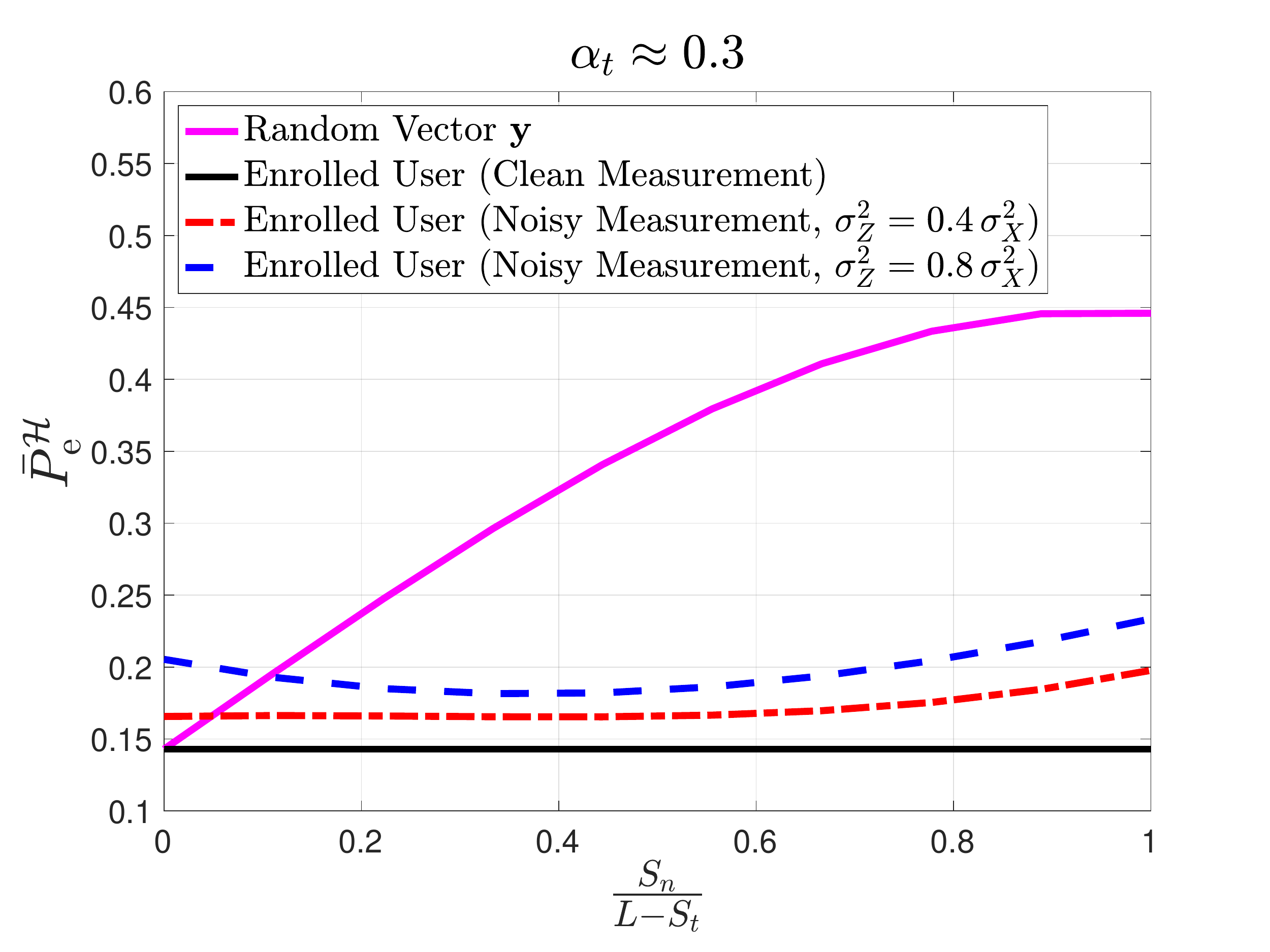}
        \vspace{-17pt}
         \caption{}
        \label{fig:AveProbebilityError3}
    \end{subfigure}
    ~ 
    \begin{subfigure}[b]{0.23\textwidth}
        \includegraphics[width=4.65cm, height=3.5cm]{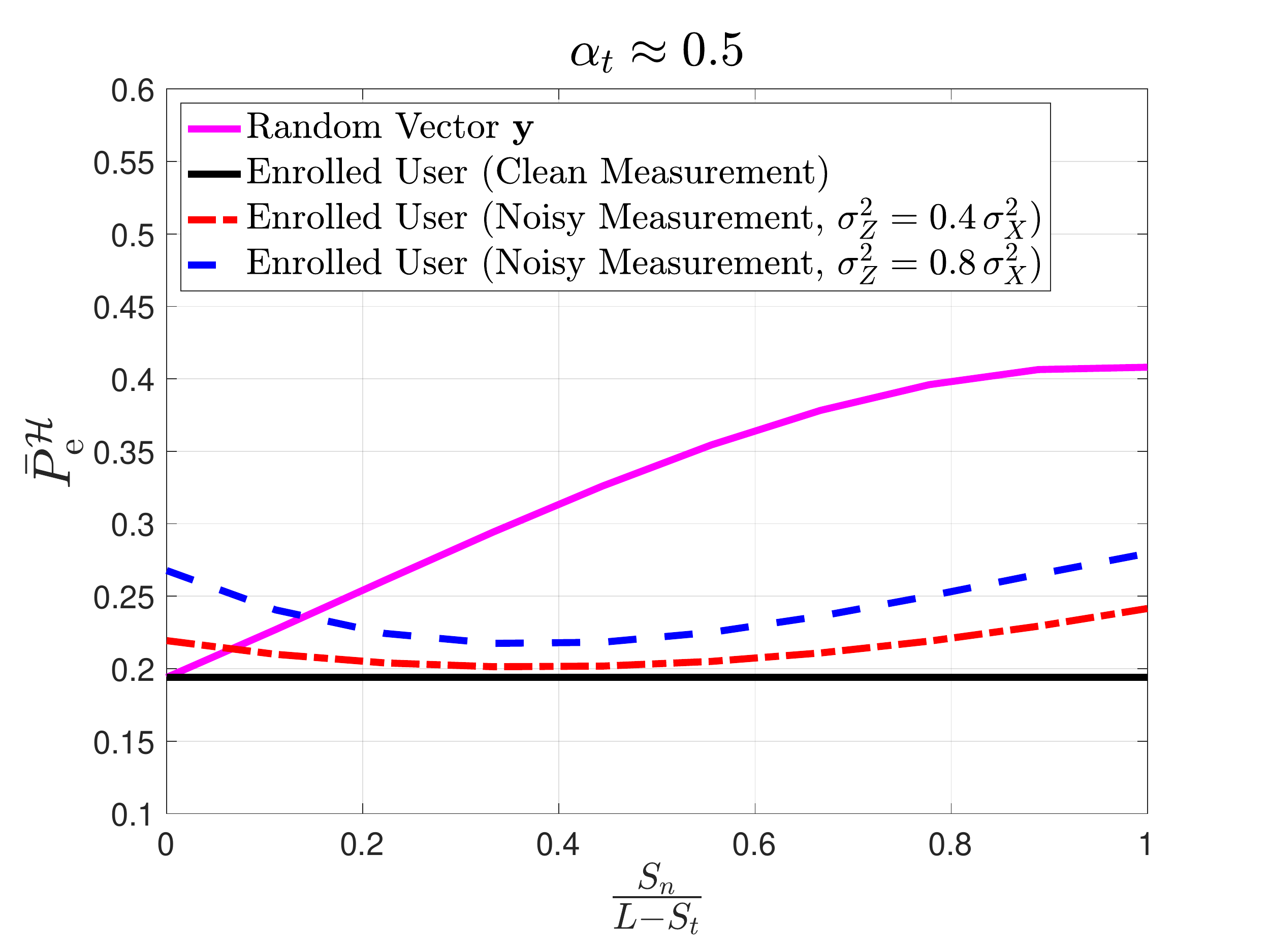}
        \vspace{-17pt}
         \caption{}
        \label{fig:AveProbebilityError4}
    \end{subfigure}
     \caption{\it Error probabilities for different sparsity ratios $\alpha_{\rm t}$ and measurement noise $\sigma^2_Z$;  setting $\sigma^2_X = 0.5$, $ \tau = \sigma_X$.}
 \label{Fig:AveProbebilityOfError}
\end{figure}

We consider a verification vector $\mathbf y\in\mathbb R^N$, which is allegedly from user~$c$.
It is either a noisy version
of the enrolled $\mathbf x(c)$ or completely unrelated to it
(but drawn from the same distribution).
The former case is referred to as Hypothesis ${\cal H}_1$, the latter as ${\cal H}_0$.
We write
$\mathbf{y}^{\rm auth}(c) = \mathbf{x}(c) + \mathbf{z}$, with Gaussian noise
$\mathbf{z} \sim \mathcal{N} \left( \mathbf{0}, \sigma_Z^2 \mathbf{I} \right)$.
The verification procedure works as follows.
The verifier computes
$\mathbf{u}(c)\!~\!\cdot~\psi^{\rm stc}_\lambda(\mathbf W(c)\mathbf y)$.
If this inner product exceeds a threshold, then he decides on ${\cal H}_1$, otherwise ${\cal H}_0$.
The expression $\psi^{\rm stc}_\lambda(\mathbf W(c)\mathbf y)$ is essentially
an STC `enrollment' of $\mathbf y$ without ambiguation noise.
Taking the inner product with $\mathbf u(c)$ is meant to remove the ambiguation noise from $\mathbf u(c)$
(we call this `purification'),
and it results in a similarity score.

The goal of the SCA mechanism is to prevent recovery of 
$\mathbf{x}$ and the corresponding $\psi(x_n)$ from the enrolment data 
while enabling a verifier to check if $\mathbf y$ is consistent with enrolment data $\mathbf u(c)$.
We can characterize the performance of SCA in terms of: 

(i) preservation of mutual information between $\mathbf X$ and $\widehat{\mathbf X}$ in the authorized case, 
i.e. 
$I ( \mathbf X ; \widehat{\mathbf X}  |  \mathcal{H}_1  ) \rightarrow H ( \mathbf X )$,
whilst in the unauthorized case
$I ( \mathbf X ; \widehat{\mathbf X}   |  \mathcal{H}_0 ) \rightarrow 0$. 
The same holds for a function of $\mathbf X$, i.e., 
$I( \psi (  X_n  ); \psi  ( \widehat{X}_n ) |  \mathcal{H}_1 )  
\rightarrow H ( \psi ( X_n) )$
and 
$I \big( \psi  ( X_n  ); \psi ( \widehat{X}_n )  |  \mathcal{H}_0 \big)  \rightarrow  0$. 

(ii) Reconstruction error: 
We investigate the error probabilities 
\begin{equation}
    P_{\mathrm{e}}^{\mathcal{H}}(c) = \Pr [ \psi^{\rm sbc}_\tau  ( \widehat X_n(c) ) 
    \neq \psi^{\rm sbc}_\tau ( X_n(c) )  | \mathcal{H} ]. 
\end{equation}
Let $p\triangleq {\rm Pr}[\psi^{\rm sbc}_\tau(X_n)=1]$, $p<\frac12$.
Ideally it should hold that $P_{\rm e}^{{\cal H}_1}\rightarrow 0$ and 
$P_{\rm e}^{{\cal H}_0}\rightarrow 2p(1-p)$. 
The latter expression is for random $\widehat X_n$ independent of $X_n$, with the same distribution.


\subsection{Performance results}
\label{Subsec:PerformanceAnalysis}

\begin{figure}[!t] 
      \centering%
      \hspace{-10pt} 
    \begin{subfigure}[b]{0.23\textwidth}
     \includegraphics[width=4.45cm, height=3.5cm]{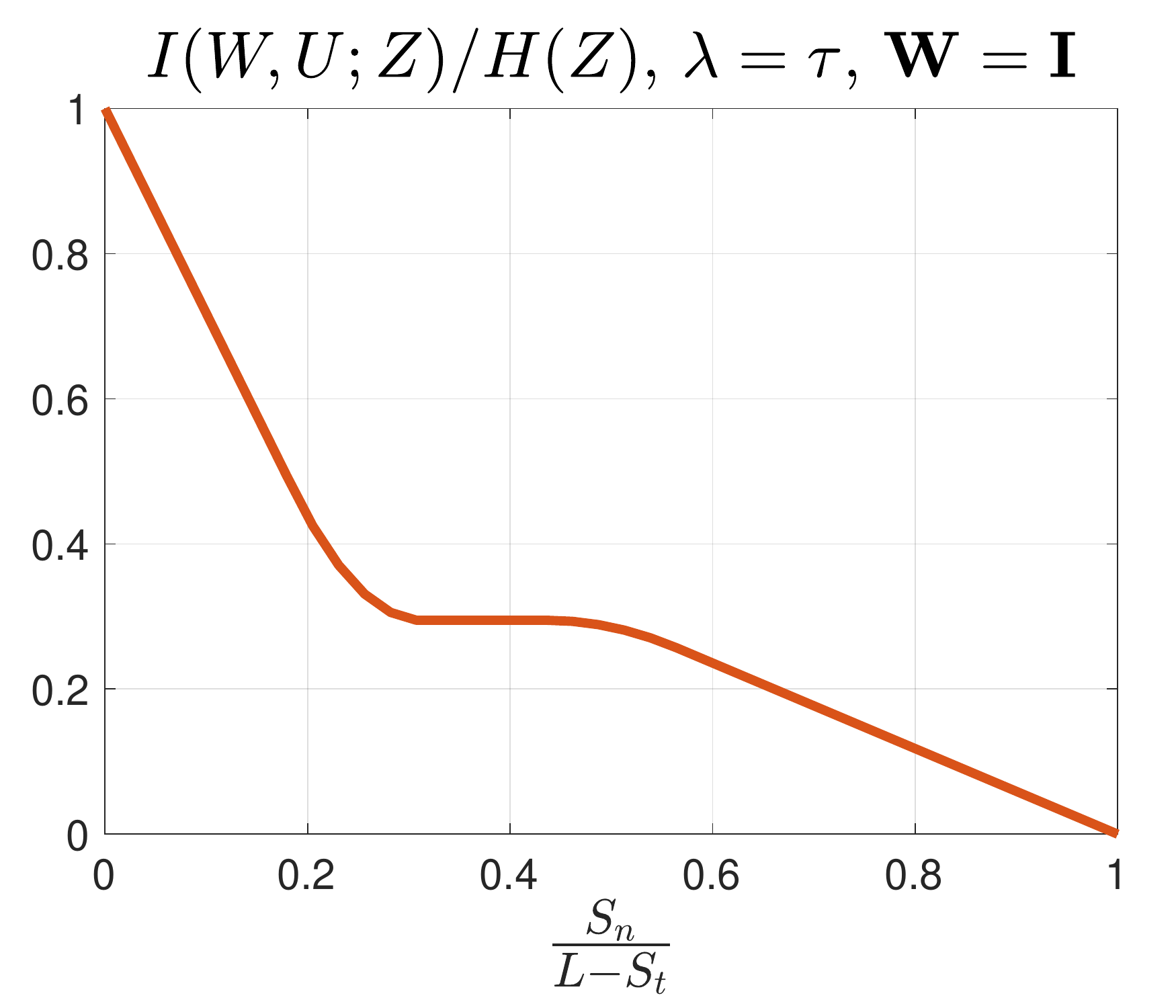}
        \vspace{-17pt}
         \caption{ }
         \label{fig:InfoLeakageMaximal}
    \end{subfigure}
    ~
    \begin{subfigure}[b]{0.23\textwidth}
        \includegraphics[width=4.45cm, height=3.5cm]{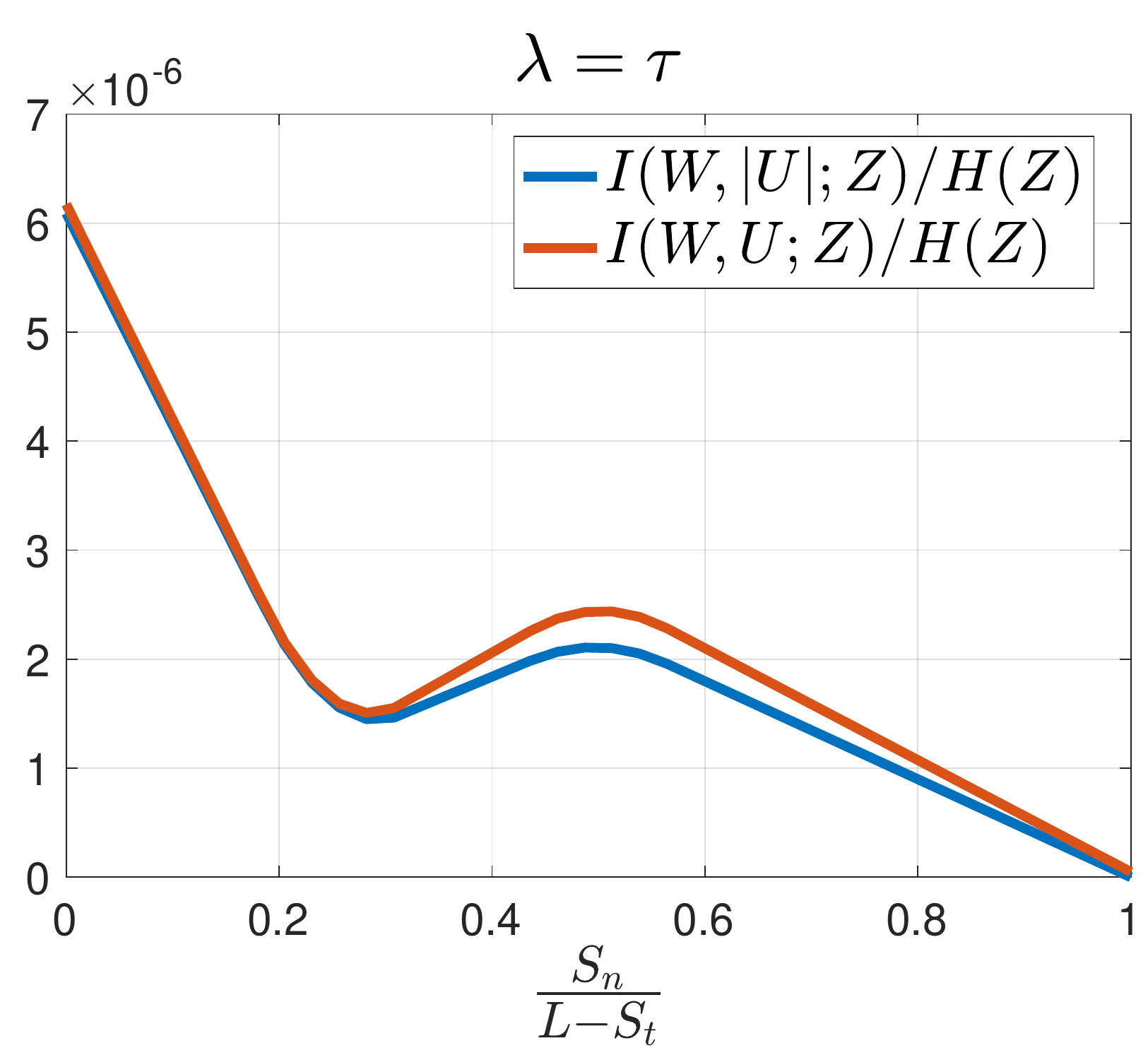}
        \vspace{-17pt}
         \caption{ }
        \label{fig:InfoLeakageWnotIdentity}
    \end{subfigure}
     \caption{\it Normalized Information Leakage $\frac{I(WU; Z)}{H(Z)}$ as a function of ambiguation ratio $\frac{S_n}{L - S_t}$. $\lambda=\tau=\sigma_X$.
     {\bf (a)} $\mathbf{W}=\mathbf I$; 
     {\bf (b)} $\mathbf{W}$ equals the PCA transform of the matrix $[ \mathbf x(1)\cdots\mathbf x(C)]$.}
 \label{Fig:InformationLeakage}
\end{figure}

We consider a database $\mathbf{X}$ of $C =100,000$ random vectors (individuals) with dimensionality $N = 256$, which are
generated from the distribution $\mathcal{N} (\mathbf{0}, \sigma^2_X \mathbf{I})$. 
We then generate the noisy version of $\mathbf{X}$ with two different 
noise variances $\sigma^2_Z = 0.4 \sigma^2_X$ and $\sigma^2_Z = 0.8 \sigma^2_X$. 
We consider square matrix $\mathbf W$, 
i.e., $L = N$.

We look at the error probability averaged over the users and the components,
$
\overline {P}_{\mathrm{e}}^{\mathcal{H}} =  \frac{1}{CN} \sum_{n=1}^N \sum_{c=1}^C P_{\mathrm{e}}^{\mathcal{H}}(c)
$.

Fig.~\ref{Fig:AveProbebilityOfError} shows the averaged error probability $\overline{P}_{e}^{\mathcal{H}}$ 
as a function of the ambiguation ratio $\frac{S_n}{L  - S_t}$, 
for fixed $\sigma_X^2 = 0.5$, $\tau = \sigma_X$ 
and different sparsity ratios $\alpha_t \triangleq S_t/L = 0.025, 0.1, 0.3, 0.5$ and measurement noise variances $\sigma_Z^2 = 0.4 \sigma_X^2, 0.8 \sigma_X^2$. 
Several things are worth noting.
\begin{itemize}
\item 
For the un-enrolled case (random $\mathbf y$), the error probability in guessing the privacy bit $Z$ increases as a function
of the ambiguation ratio. This is as expected.
\item
In the genuine user case the situation is more complex; the ambiguation noise interferes with the 
measurement noise.
\item
There is a clear gap between the genuine user case and the un-enrolled case.
The low\footnote{
Note that the False Negative probability for the overall user matching is much lower than
the single-component reconstruction error.
} (and in some plots nearly constant) error rate for genuine users demonstrates that the `purification' correctly removes
the ambiguation noise. 
\end{itemize}

Furthermore
we compute the leakage  $I ( \mathbf W , \mathbf U ; Z ) / H(Z)$. 
Fig.\,\ref{fig:InfoLeakageMaximal}
shows what happens when $\mathbf W$ is set to the trivial value $\mathbf W=\mathbf I$.
The leakage decreases from 100\% to zero with increasing ambiguation.
The curve seems to consist of three ambiguation ratio regimes, with piecewise linear behaviour:
$0-0.25$, $0.25-05$ and $0.5-1$.
At the moment we are not able to explain this behaviour.

Fig.\,\ref{fig:InfoLeakageWnotIdentity} shows what happens when a less trivial matrix $\mathbf W$
is used, namely the PCA transform matrix of the matrix $[\mathbf x(1)\cdots \mathbf x(C)]$.
This same $\mathbf W$ is used for all users.
We again observe piecewise linear behaviour with the same three intervals.
However, the middle piece is no longer constant but increasing.
More importantly, the leakage is reduced by orders of magnitude.
Fig.\,\ref{fig:InfoLeakageWnotIdentity} also shows the leakage from Sparse Binary Coding with ambiguation;
it is slightly smaller than for the ternary case.

\subsection{Non-square projection matrix}
\label{sec:nonsquare}
We briefly discuss the case $L>N$, i.e. the number of random projections is larger than
the dimension of $\mathbf x$.
The adversary is confronted with an ambiguized ternary vector $\mathbf u$,
and from it has to guess the $\mathbf v=\psi^{\rm stc}_\lambda(\mathbf{Wx})$
by guessing which locations in $\mathbf u$ contain the ambiguation noise.
When $L$ is larger than $N$, the adversary may be able to distinguish between wrong guesses and the correct guess,
as follows.
For a wrong guess it will typically hold that $WW^\dag \mathbf v_{\rm wrong}$
is far away from $\mathbf v_{\rm wrong}$, while on the other hand
it holds that $WW^\dag \mathbf v\approx \mathbf v$.
From the correct $\mathbf v$ an estimator for $\mathbf x$ is then obtained as
$\widehat{\mathbf x}=W^\dag\mathbf v$.
Hence, information-theoretically speaking, there is no privacy protection.
However, the amount of effort in going through all the possible guesses scales as
${S_t+S_n \choose S_n}$, which is huge.
The security is {\em computational}, not information-theoretic.


\section{Discussion}
\label{sec:Conclusion}

For quantizing HDSs we have established that there is a clearly identifiable optimal choice
for protecting the $V$ and $Z$ bits:
taking the number of quantization intervals to be even, and setting $m=2$.
However, for noise tolerance it is advantageous to set $m$ as large as possible.
The $Z$-leakage result for $m\to\infty$ (Fig.\,\ref{fig:leakHDS})
has some caveats. 
It is nice that a minimum exists at $F(-\tau)=p_0/2$ and $F(-\tau)=p_0$,
but unfortunately the operational meaning of $\tau$ is not really well defined.
A small shift of $\tau$ has little impact on the concept 
``this variable is abnormally far from zero'', but has a large effect in Fig.\,\ref{fig:leakHDS}.
It is left as a topic for future work to study this further.

There are no such subtleties for the Code Offset Method.
We think we can safely conclude that the COM has only negligible leakage.

For the SCA approach we have 
established that there is a clear gap between how much you know about $Z$ 
if you {\em do} and {\em do not} have access to a matching verification
measurement $\mathbf y(c)$. (Not having such access means trying to reconstruct $Z$ from the public data.)
This is visible as a gap (Fig.\,\ref{Fig:AveProbebilityOfError}) in the error probability for reconstructing~$Z$,
and as low mutual information $I(\mathbf{W}, \mathbf{U};Z)$ in Fig.\,\ref{fig:InfoLeakageWnotIdentity}.
Determining the leakage about sign($X_n$) is left for future work.
Other topics for future work are further experimentation with different choices of the projection matrix $\mathbf W$
and understanding the piecewise linear shape of the leakage curve.

\bibliographystyle{IEEEtran}
\bibliography{references}

\end{document}